\documentclass[conference]{IEEEtran}
\usepackage{xspace, tabularx}
\usepackage{url}
\usepackage[
bookmarksopen,
bookmarksdepth=2,
colorlinks=false,
urlcolor=blue]{hyperref}
\usepackage{color}
\usepackage{soul}
\usepackage{cite}
\usepackage[margin=4pt,caption=false]{subfig}
\usepackage{diagbox}
\usepackage{multirow}
\usepackage{algorithm}
\usepackage{graphicx}
\usepackage{amsmath, bm}
\usepackage{algpseudocode}
\usepackage{xcolor}
\usepackage{tikz}
\usetikzlibrary{positioning, shapes.geometric}
\usepackage{colortbl}
\usepackage{multirow}

\usepackage{amsthm}
\theoremstyle{plain}
\newtheorem{definition}{Definition}[section]
\newtheorem{theorem}{Theorem}[section]

\usepackage{comment}
\usepackage[shortlabels]{enumitem}

\usepackage{dirtytalk}

\usepackage{amssymb}
\usepackage{alltt}
\usepackage{comment}
\usepackage{makecell}

\allowdisplaybreaks
\usepackage{xparse}

\usepackage[utf8]{inputenc}
\usepackage[english]{babel}
\definecolor{lavender}{rgb}{0.75, 0.58, 0.89}

\hypersetup{draft}
\begin{document}

\title{Is Stellar As Secure As You Think?}

\author{\IEEEauthorblockN{Minjeong Kim}
\IEEEauthorblockA{KAIST\\
mjkim9334@kaist.ac.kr}
\and
\IEEEauthorblockN{Yujin Kwon}
\IEEEauthorblockA{KAIST\\
dbwls8724@kaist.ac.kr}
\and
\IEEEauthorblockN{Yongdae Kim}
\IEEEauthorblockA{KAIST\\
yongdaek@kaist.ac.kr}}

\maketitle

\begin{abstract}
Stellar is one of the top ten cryptocurrencies in terms of market capitalization.
It adopts a variant of Byzantine fault tolerance (BFT), named federated Byzantine agreement (FBA), which generalizes the traditional BFT algorithm to make it more suitable for open-membership blockchains. 
To this end, FBA introduces a concept called \textit{quorum slice}, which consists of a set of nodes. 
In FBA, a node can complete one consensus round when it receives specific messages from nodes in a quorum slice appointed by the node.
In this study, we analyze FBA, whose security is highly dependent on the structure of the quorum slices, and demonstrate that it is not superior to the traditional BFT algorithm in terms of \textit{safety} and \textit{liveness}.
Then, to analyze the security of the Stellar consensus protocol (SCP), which is a construction for FBA, we investigate the current quorum slices in Stellar. 
We analyze the structure of quorum slices and measure the influence of each node quantitatively using two metrics, PageRank (PR) and the newly proposed NodeRank (NR).
The results show that the Stellar system is significantly centralized. 
Thereafter, to determine how the centralized structure can have a negative impact on the Stellar system, we study the \textit{cascading failure} caused by deleting only a few nodes (i.e., validators) in Stellar.
We show that all of the nodes in Stellar cannot run SCP if only two nodes fail.
To make matters worse, these two nodes are run and controlled by a single organization, the Stellar foundation.


\end{abstract}


\IEEEpeerreviewmaketitle

\section{Introduction}
\label{sec:intro}

Currently, many cryptocurrencies based on peer-to-peer networks use open ledgers called \textit{blockchains}~\cite{bonneau2015sok}, on which transactions are publicly recorded. 
Nodes that validate a set of transactions that are going to be recorded in the blockchain are called \textit{validators.} 
To make an agreement among validators, a system has a consensus mechanism.
Nowadays, many blockchain projects utilize practical Byzantine fault tolerance (PBFT)~\cite{castro1999practical} as a consensus mechanism because it has some advantages, such as a high transaction throughput and no waste of energy~\cite{neo, ziliqa, miller2016honey,armknecht2015ripple,cachin2016architecture}.
Owing to these advantages, Stellar also uses a variant of PBFT~\cite{castro1999practical}, called federated Byzantine agreement (FBA)~\cite{stellar}.
However, because PBFT requires to fix the members participating in the consensus process in advance, it is unsuitable for a blockchain with open membership, where anyone can join in and out at any time.\footnote{The Stellar whitepaper~\cite{stellar} never uses the term ``permissionless blockchain". Instead, it uses the term ``open membership", but it is difficult to distinguish the difference of the meaning between the two terms.} 
Therefore, to overcome this limitation of PBFT, FBA provides an open membership service based on a \textit{trust model}, where anyone can join the network, but only nodes trusted by others can be validators.
More specifically, in FBA, each node selects a set of nodes that it trusts, which is called a \textit{quorum slice}.
To make an individual decision on a given statement, it needs to receive specific messages from the members of its quorum slice.
A \textit{quorum} is an union set of all slices for each member in it, which leads to a global consensus on a given statement.
Since the behavior of members in a quorum slice affects the node that creates it, the quorum slice is an important component of the Stellar security.
However, to the best of our knowledge, there has been no study on the structure of quorum slices in the current Stellar system and how a poor structure can affect the system.
In this study, for the first time, we 1) analyze FBA in terms of safety and liveness, 2) measure and analyze the structure of Stellar quorum slices, and 3) evaluate the robustness of the current Stellar system against cascading failure. 
In the rest of this section, we briefly describe each of our contributions.


\noindent\textbf{FBA Analysis. }
For the FBA analysis, the term \textit{$(f,x)$-fault tolerant (FT) system} describes how vulnerable the system is to a node failure.
The variable $f$ and $x$ indicate the minimum number and minimum fraction of faulty nodes that a system cannot tolerate, respectively.
Therefore, a $(f,x)$-FT system implies that if more than $x$\% $(=\frac{100\cdot f}{N})$ of all active validators are faulty, where $N$ indicates the total number of active validators in the system, then the system will not reach any consensus.
Then, we find that the values of $x$ and $f$ in a FBA-based system depend on the structure of quorum slices, and $x$ ranges from 0 to $\frac{100}{3}$. 
From this, we conclude that the value of $x$ in FBA is always lower than or equal to that in PBFT, because the latter achieves the maximum value of $x$ (i.e. $\frac{100}{3}$).

\noindent\textbf{Data Analysis. }
Next, to determine the values of $x$ and $f$ in Stellar, we study the current structure of quorum slices in it.
To this end, we first collect data for existing quorum slices in the Stellar system. 
With these data, we analyze various characteristics of quorum slices and measure the influence of each node in quorum slices quantitatively. 
In this process, we use two metrics, PageRank (PR) and NodeRank (NR), where NR is newly proposed to better reflect the concept of a quorum slice. 
Based on this analysis, we find two issues in the current Stellar system: 
1) the size of the quorum slices is small, and 2) the influence of each node is significantly biased.
As a result, we observe that the structure of quorum slices is significantly centralized.

\noindent
\textbf{Cascading Failure. }
Finally, we study cascading failures in the centralized Stellar system.
According to our results, all nodes get stuck after only two nodes fail.
This implies that, currently (Jan. 2019), the value of $x$ in Stellar is approximately 4.5 ($\approx\frac{2}{44}\times 100$), where the total number of active validators is 44 and $f$ is 2. 
To make matters worse, the two nodes that cause the entire system to be stuck are run by the same organization, namely the Stellar foundation, making Stellar vulnerable to ``single point of failures''. 

In summary, our contributions are as follows:

\begin{enumerate}
\item We analyze FBA and prove that it is not superior to PBFT in terms of safety and liveness.
\item We conduct a data analysis on the Stellar system, and show that the structure of quorum slices is highly centralized.
\item We study cascading failures considering the current quorum slices. Our results imply that validators cannot achieve a consensus after deleting only two nodes run by the Stellar foundation. 
\end{enumerate}

The remainder of this paper is organized as follows.
In Section~\ref{sec:background}, we present the background of Stellar.
The analysis of FBA and the Stellar system are described in Sections~\ref{sec:theoretical} and \ref{sec:analysis}, respectively.
In Section~\ref{sec:cascade}, we study cascading failures of Stellar.
In Section~\ref{sec:discuss}, we discuss the possible ways to reduce the impact of cascading failures and their limitations.
We describe some related works in Section~\ref{sec:related} and 
conclude in Section~\ref{sec:conclusion}.
\section{Background on Stellar}
\label{sec:background}

Stellar is an open platform based on a blockchain.
It is designed for providing the fast and low cost payment service, including cross-border transactions. 
In this section, we review the main features of Stellar relevant to this study.

\noindent
\textbf{Federated Byzantine agreement (FBA). }
FBA, on which Stellar is based, generalizes Byzantine agreement (BA) to express a greater range of settings.
Therefore, PBFT~\cite{castro1999practical}, one of the best known variants of BA, can be explained by FBA.
Note that BA is an algorithm that guarantees a consensus despite the existence of some Byzantine nodes.
The main difference between FBA and traditional BA (or non-federated BA) is whether they will support open membership, allowing any node to participate in the consensus process if desired.
In the traditional BA, there is no open membership, which must be set by a central authority in advance, to avoid Sybil attacks~\cite{douceur2002sybil}.
On the other hand, in FBA, anyone can have a chance to be in the membership list at any time as long as it is trusted by others. More specifically, any node that has been being included in at least one quorum slice can be in the membership list.

In addition, to make nodes reach a consensus, SCP should provide the following:

\begin{itemize}
\item \textit{safety.} A set of nodes satisfies \textit{safety} if no two of them ever reach an agreement on different values at the same time.

\item \textit{liveness.} A node satisfies \textit{liveness} if it can reach an agreement on a new value.
\end{itemize}

The quorum slice and quorum, key elements of FBA, are explained in detail in the rest of this section. 

\noindent
\textbf{Quorum Slice/Quorum. }
A node in the system must select a set of nodes called \textit{quorum slice}, each consisting of nodes it trusts,\footnote{The terms quorum slice and quorum sets~\cite{crafting_quorum_set} are equivalent. In addition, nodes in quorum slices are based on trust that already has exists in real life, such as bank or financial institution.} and then receive specific messages from them to make an individual decision on a given statement.
A \textit{quorum} is a set of nodes that contains at least one quorum slice for each of its members.
Therefore, a quorum is a set of nodes sufficiently large to reach a consensus within a system and a quorum slice is a subset of a quorum that directly determines a node's consensus.

In fact, a node can belong to multiple quorum slices, and a quorum slice may contain another quorum slice as a member, called a \textit{nested quorum slice}.
Moreover, every quorum slice has a threshold value, which can be different for each slice.
If the number of nodes above the threshold in the quorum slice agrees on the same statement, then the node that has selected the slice also agrees on it. 
Therefore, a quorum slice formed carelessly can cause a failure of the consensus.
To prevent this, each quorum has to satisfy two quorum formation conditions.
The first is that any two quorums should have an intersection even after deleting malicious nodes in the quorums.
This implies that an overlapping portion, consisting of non-Byzantine nodes, must exist. 
The second condition is that a quorum still exists after deleting Byzantine nodes.

These two conditions guarantee the system's safety and liveness. 
However, because the quorum slice structure depends on user-configurable parameters, it is not guaranteed that users will always form quorum slices that satisfy those two conditions.

\noindent
\textbf{Types of node. }
Nodes in the Stellar system can be divided into two types based on their properties:
\begin{itemize}
    \item \textit{well-behaved nodes.} These choose acceptable quorum slices and respond normally to all requests from peers.
    \item \textit{ill-behaved nodes.} These suffer from Byzantine (i.e., acting arbitrarily) or crash (i.e., stopping responding to requests or halting) failures.
\end{itemize}
Thus, an ill-behaved node can send a fake message to others, or fail to send a message at all.
Note that a well-behaved node may not work properly due to the effect of an ill-behaved node, such as waiting endlessly to receive messages.

\section{FBA Analysis}
\label{sec:theoretical}

In FBA, we can classify nodes into three groups according to their properties.
The first group \textit{A} includes ill-behaved nodes, and the second group \textit{B} includes well behaved nodes that cannot work properly because they are affected by ill-behaved nodes (i.e., those in group \textit{A}).
The third group \textit{C} includes well-behaved nodes that are not affected by ill-behaved ones, and thus, work properly.

When considering the three groups, SCP, which is a construction for FBA, only guarantees the correctness of \textit{C} under the assumption that nodes in \textit{C} satisfy the quorum formation conditions, whereas the safety and liveness of groups \textit{A} and \textit{B} are not guaranteed. 
Nevertheless, if most of the well-behaved nodes in Stellar always belong to \textit{C}, the system can be regarded as secure.

However, whether most nodes can still belong to \textit{C} after the failure of a few nodes depends significantly on the structure of the quorum slices.
This is because even with the same set of nodes in \textit{A}, the number of nodes in \textit{B} could change by the structure.
For example, assume that a validator named \textit{Alice} is the only ill-behaved node.
If every quorum slice in Stellar consists solely of \textit{Alice}, then they will all be affected by \textit{Alice}'s failure. 
In other words, once \textit{Alice} goes to \textit{A}, all the remaining nodes would belong to \textit{B} because their consensus process is blocked due to the lack of messages from \textit{Alice}.
Therefore, because there is no node in \textit{C}, the safety and liveness for all nodes are not guaranteed.

On the other hand, when all nodes evenly select members in their slices, most of the remaining nodes are still in the group \textit{C} even after the failure of \textit{Alice} node.
In this case, the safety and liveness of the system are much stronger than in the previous case.
This implies that the structure of quorum slices can significantly affect the safety and liveness of FBA.

In fact, because in Stellar, users select trusted validators in their slices manually, the structure of quorum slices may be significantly unsafe, depending on users' choices.
Therefore, it is possible for the system to have an inappropriate structure of quorum slices where most nodes are in group \textit{B} rather than in \textit{C}, even if the size of \textit{A} is small. 

To quantitatively measure the extent to which the structure of quorum slices can affect the system, we define a $(f, x)$-FT system, representing how much the system is tolerant to ill-behaved nodes (i.e., the nodes in group \textit{A}).

\begin{definition}[$(f, x)$-FT system]
\label{x-crushed}
In $(f, x)$-FT systems, all nodes can eventually agree on the same value that does not contradict the history of a consensus process, if less than $f$ nodes are ill-behaved, accounting for $x\%$ of the total active validators.
\end{definition} 

Note that the value of $x$ in FBA can be changed depending on the structure of the quorum slices.
The smaller the values of $f$ and $x$, the less tolerant the system is to ill-behaved nodes.
The following theorem shows the range of $x$ in FBA. 

\begin{theorem}
\label{thm:max}
The value of $x$ in FBA ranges from 0 to $\frac{100}{3}.$
\end{theorem}

\begin{proof}
Let $t_s$ and $t_l$ be the thresholds for safety and liveness, respectively, in an asynchronous deterministic system~\cite{castro1999practical}.
Then, the system is safe if less than $t_s$ nodes are ill-behaved.
Likewise, the system can exhibit liveness if less than $t_l$ nodes are ill-behaved.

First, we obtain the maximum value of $x$ as $\frac{\min (t_s,t_l)}{N}\times 100$, where $N$ is the total number of validators in the system. 
Note that $N-t_s-t_l\geq t_s$ should met, because the number of nodes sending a correct message should be larger than or equal to the number of nodes sending an incorrect message\footnote{Note that ill-behaved node is defined in Section~\ref{sec:background}.}. 
Therefore, for the maximum value of $x$, it is sufficient to calculate $\frac{\min (t_s,t_l)}{N}\times 100$, considering $N\geq 2t_s+t_l.$ 
To do this, we consider two cases: 1) $t_s< t_l$ and 2)  $t_s\geq t_l.$
In the first case, the following is satisfied. 
\begin{equation*}
\max\{\min (t_s,t_l)\}=\max t_s=\max \frac{N-t_l}{2}    
\end{equation*}
Because $\frac{N-t_l}{2}$ is less than $t_l,$ the maximum value of $t_s$ is $\frac{N}{3}$ according to the below equation. 
\begin{equation*}
    \frac{N-t_l}{2}<t_l \Longleftrightarrow \frac{N}{3}<t_l,\,\,t_s<\frac{N}{3} 
\end{equation*}
As a result, the maximum value of $x$ is $\frac{100}{3}.$
In the second case, the below equation is met.  
\begin{equation*}
\max\{\min (t_s,t_l)\}=\max t_l=\max (N-2t_s)    
\end{equation*}
Similar to the first case, we can find out that the maximum value of $x$ is $\frac{100}{3}.$ 

Next, we prove that the minimum value of $x$ is 0. 
It is sufficient to give an example where $x=0.$
We consider the case where there are a sufficiently large number of validators and every slice only contains the same three nodes. In this case, when those three nodes fail, the entire system would be stuck. 
Therefore, for a sufficiently large $N,$ $x$ is close to 0, which completes the proof. 
\end{proof}

The above theorem states the maximum value of $x$ is $\frac{100}{3}.$ Note that the value of $x$ in PBFT is $\frac{100}{3}$ because it can guarantee safety and liveness until less than $\frac{1}{3}$ fraction of all nodes are Byzantine~\cite{castro1999practical}. 
Therefore, Theorem~\ref{thm:max} implies that FBA is not superior to PBFT in terms of safety and liveness.
As an extreme case, if the structure of quorum slices is completely centralized (i.e., every node selects only one common validator in its slice), the system collapses with a failure of only one node, which implies that $x$ is close to 0.

\section{Data Analysis}
\label{sec:analysis}
In this section, we determine the characteristics of the quorum slices in the current Stellar system and analyze the nodes' influence quantitatively. 
Further, through this analysis, we confirm that the Stellar system is significantly centralized. 
In other words, in the current Stellar system, 1) each quorum slice contains only a few validators, and 2) the influence among validators is highly biased.

\subsection{Characteristics of Quorum Slices}
\label{subsec:numberOfValidators}

\begin{figure*}[t]
\centering
 \subfloat[Total number of validators and quorum slices in the system over time]{
\includegraphics[width=0.4\textwidth]{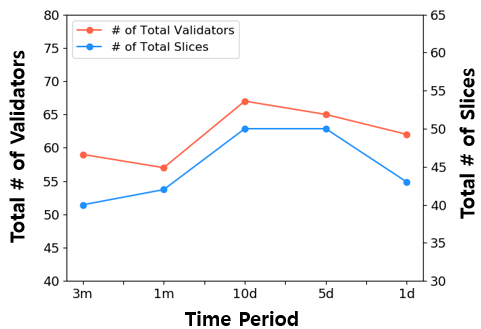}
\hspace{1cm}
 \label{fig:numberOfValidators}
 }
  \subfloat[Number of validators in each quorum slice]{
\includegraphics[width=0.37\textwidth]{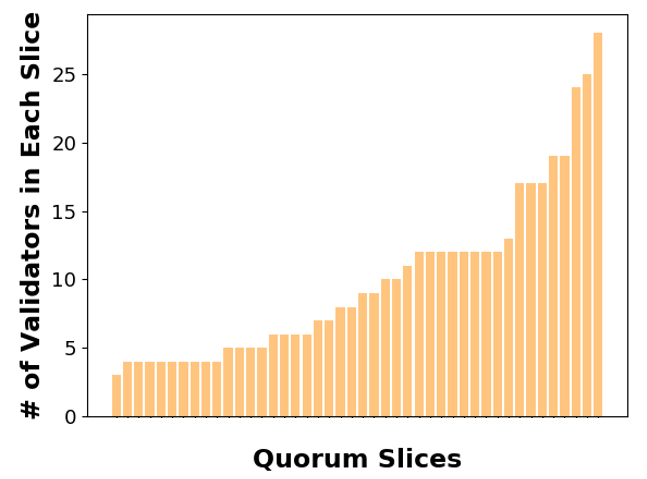}
 \label{fig:nodesInSlice}
 }
\caption{Statistics of quorum slices and validators}
\end{figure*}

\smallskip\noindent
\textbf{Methodology. }
We join multiple nodes in the Stellar system as validators (i.e., full nodes).
Subsequently, we collect data on the Stellar system through these validators.
The data includes the validators' Internet Protocol (IP) addresses, the generated quorum slices, and a list of existing validators.
For the data reliability, we also collect the corresponding data from several websites, such as mystellar.tools~\cite{mystellar}, stellarbeat.io~\cite{stellarbeat.io}, and stellar.org~\cite{stellar.org}, which provide information on the quorum slices and validators.

In Section~\ref{subsec:numberOfValidators}, we first measure the number of total validators and quorum slices in the long and short terms to investigate the characteristics of the quorum slices in the system.
Subsequently, we measure the number of nodes in each existing slice.
To this end, we analyze data over the three months, one month, ten days, five days, and one day before January 29, 2019. 

\smallskip\noindent
\textbf{Results and Implications. }
The results are presented in Fig.~\ref{fig:numberOfValidators}, where
the red and blue lines indicate the number of validators (in the left axis) and the number of quorum slices (in the right axis), respectively. 
As shown in Fig.~\ref{fig:numberOfValidators},
the number of validators during the period typically ranges from 57 to 67, and the number of existing quorum slices ranges from 40 to 50.
Moreover, only 31 validators remain active continuously throughout the whole study period.
This implies that there are significantly small numbers of validators and slices in the Stellar system, and to make matters worse, only half of the existing validators are active continuously.
The others are just temporary validators that easily churn in and out of the system.

Fig.~\ref{fig:nodesInSlice} represents the result of the number of nodes in each existing slice.
It shows that 58\% of the quorum slices have ten or fewer validators.

\smallskip\noindent
\textbf{Causal Analysis. }
The first reason for the small number of validators is the insufficient incentivization to make users participate in Stellar as validators. 
Unlike many other blockchain systems that give internal incentive to nodes participating in a consensus algorithms (e.g., financial rewards such as mining reward in the proof-of-work  system), Stellar does not.
\begin{table}[t]
\centering
\caption{Motivations of players participating in Stellar as validators.}
\begin{tabular}{|c|c|c|c|}
\hline
Purpose      &Type       & Number of players    & Rate (\%)  \\
\hline \hline
\multirow{2}{*}{For-profit}   &Business with Stellar   & 23    & 74.2\\ \cline{2-4}
    &Stellar Foundation          & 3            & 9.7  \\  \hline
Non-profit      & Individual            & 1             & 3.2   \\\hline
\multicolumn{2}{|c|}{Unknown}            & 4             & 12.9  \\
\hline
\end{tabular}
\label{table:purposeOfValidator}
\end{table}

To investigate why the current validators are present without such incentivization, we analyze the 31 validators that were continuously active during the studied period.
Note that the others between 14 and 36 validators are just temporary validators. 
Fortunately, those 31 active validators all have exposed their identities (i.e., name, contact address, node ID and description) publicly in the website~\cite{validator_list} because they need to advertise themselves in order to be trusted by others.
Therefore, with the data, we observe the characteristics of the 31 validators participating in the Stellar system.

These observations are shown in Table~\ref{table:purposeOfValidator}, which classifies the relationships between validators and Stellar into four types.
The first type indicates a player that has a business with Stellar. 
For example, \textit{IBM} is providing a blockchain-based payment system with Stellar and runs nine validators.
Some organizations, such as \textit{satoshiPay} or \textit{tempo.eu.com} are in partnerships with Stellar for building applications on the Stellar platform.
The second type is controlled by the Stellar foundation itself, such as \textit{sdf\_validator1}.
The third type is a node (run by an individual) that participates in the system for better decentralization of Stellar, rather than for profit.
Furthermore, despite the website's~\cite{validator_list} information, there are still some nodes whose identities or types are not clear, so this kind of nodes is classified into the fourth type.
From Table~\ref{table:purposeOfValidator}, we can see that 83.9\% of validators (including types one and two) are directly related to Stellar, 3.2\% are non-profit, and the rest is unknown.
This implies that only a small number of specific nodes that can receive external profits using Stellar are motivated to work as validators.
Therefore, if there is an internal profit, nodes that do not receive external profits will also be motivated to work as validators, leading to a much larger number of validators.

In addition, the second reason for the small number of validators (not only in the system but also in each quorum slice) is that Stellar is based on a trust model.
According to the structure of the current quorum slices, participants in Stellar network trust only a few organizations, i.e., those that are significantly popular or highly relevant to them.
For example, \textit{satoshiPay} whose quorum slice consists of \textit{sdf\_validator1}, \textit{sdf\_validator2}, \textit{sdf\_validator3} and \textit{eno}, is a micro-payment solution company collaborating with Stellar.
Considering that \textit{sdf\_validator} nodes are run by the Stellar foundation, one can see that the quorum slice of \textit{satoshiPay} consists of mostly nodes that are in partnership with Stellar.
This indicates that a node tends to choose the validators to which it is related.
Therefore, the trust model inevitably results in fewer nodes in the slices.

\subsection{Node Influence }
\label{subsec:AnalysisBiased}

Next, to determine what extent of the power in Stellar is biased, we first visualize the structure of the quorum slices. Subsequently, we measure the influence of each node quantitatively.

\begin{figure}[h]
\centering
\includegraphics[width=\columnwidth]{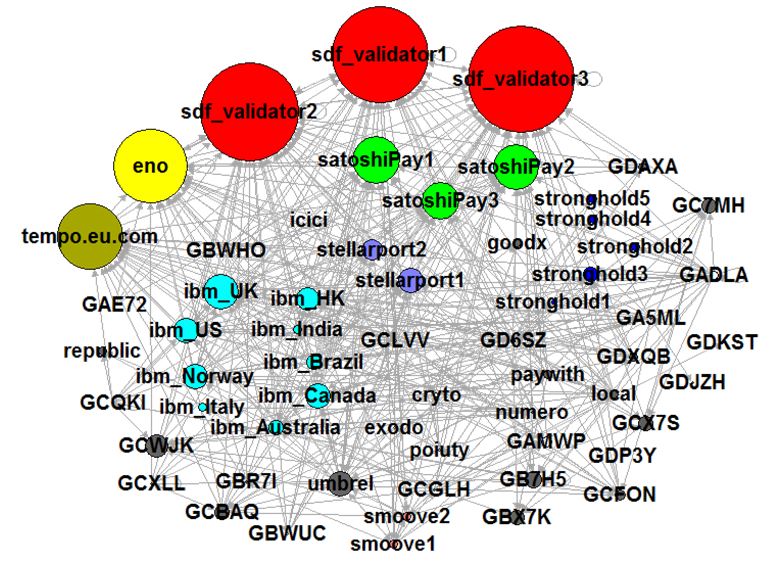}
\caption{Directed graph of quorum slices in Stellar}
\label{fig:directed_graph}
\end{figure}

\smallskip\noindent
\textbf{Visualization of Quorum Slices. }
Fig.~\ref{fig:directed_graph} shows the structure of the quorum slices on January 22, 2019, as a directed graph.
In this figure, each circle represents a validator, where the size of the circle is proportional to the number of incoming edges on each node.
In fact, the edges in the directed graph are divided into incoming and outgoing, depending on the direction.
As an example, for two validators $A$ and $B,$ if $A$ generates its quorum slice including $B,$ an edge exists from $A$ to $B$ and is referred to as the outgoing edge of $A$ and the incoming edge of $B.$ 
In particular, the number of incoming edges of a node can be indicative of the influence of a node because it is proportional to the trust it bears.
Furthermore, vertexes with the same color are run by the same organization.

From Fig.~\ref{fig:directed_graph}, we can visually see which nodes are more influential (i.e., the one with more incoming edges).
Accordingly, \textit{sdf\_validator} nodes have the most incoming edges, followed by \textit{eno} and \textit{tempo.eu.com}. 
In addition, there are a total of 62 validators run by approximately 37 different organizations.

\smallskip\noindent
\textbf{Evaluation of Node Influence. }
Next, we analyze the influence of each validator quantitatively using two metrics: PR and the newly proposed NR.

In a directed graph, the influence of a node usually depends on both the number of incoming edges and the weight of the edges. 
Therefore, to consider the weight of the edges, we adopt PR~\cite{page1998pagerank}, which is typically used to rank websites in the Google search engine.
Using PR, we can assign a higher score to a node that contains many incoming edges from influential nodes. 

However, considering only the weight of each edge is not sufficient to measure the influence of each node in Stellar accurately. 
A quorum slice, which affects the node influence, consists of a validator, a nested quorum slice, and threshold.
Thus, to measure the influence of one node in the FBA model, we need to consider the following three elements:
\begin{enumerate} 
\item The number of quorum slices that contain the node. 
\item Whether an influential node chooses the node in its slice.
\item The value of the threshold (high or low) of the slice containing the node. 
\end{enumerate}
While PR only considers 1) and 2), the proposed NR considers all the three elements. 

In Stellar, each quorum slice has a different power depending on who creates it.
For instance, when an IBM's node is influential (i.e., many nodes have the node as a member of their slices), the quorum slice of the IBM node has more power than a lesser-known slice.
Accordingly, for the elements 1) and 2), NR considers the weight of each slice as PR of the slice's generator, who creates the slice.
If there are several generators for one slice, the weight of the slice is the sum of the generators' PR.

For the element 3), we give an example below for easier understanding. Assume that a quorum slice of node $n_1$ consists of nodes $n_1$, $n_2$, and $n_3$ with threshold 3, while another quorum slice of node $n_4$ consists of nodes $n_4$, $n_5$, and $n_6$ with threshold 2. 
If node $n_2$ fails, node $n_1$ cannot achieve the consensus, because $n_1$ cannot receive the messages above threshold 3 from the members of its quorum slice.
Meanwhile, even if node $n_5$ fails, node $n_4$ can still reach the consensus, because its quorum slice has threshold 2. 
Thus, the influences of $n_2$ and $n_5$ are different, depending on the threshold of the respective quorum slices. 
More specifically, node $n_2$ is more influential than node $n_5.$ 
In general, the higher the threshold of the quorum slice containing the node, the greater the influence of the node.
%

Next, we define NR of node $v$ as follows: 
\begin{equation}
    \begin{aligned}
    a_0 (Q,v)&=1,\quad a_k (Q,v)=\frac{T_Q}{|Q|}\times a_{k-1}(Q,v)\\
    \text{NR}&=\sum_{Q \in \bm{Q}_v} \sum_{G \in \bm{G}_Q} \text{PR}_{G}\times a_{l(Q,v)} (Q,v)
    \label{eq:nr}
    \end{aligned}
\end{equation}
In Eq.~\eqref{eq:nr}, the following notations are used; 
the set of all quorum slices that include node $v$ is denoted by $\bm{Q}_v,$ and $T_Q$ is a threshold of quorum slice $Q.$ 
The number of members of $Q$ is denoted by $|Q|.$ 
The parameter $l(Q,v)$ is how many times $Q$ that includes node $v$ is nested. Therefore, if $Q=\{n_1, \{n_2\}, \{\{n_3\}\}\},$ then $l(Q,n_1), l(Q,n_2),$ and $l(Q,n_3)$ are 1,2, and 3, respectively. 
In addition, $\bm{G}_Q$ is the set of nodes (i.e., generators) that create the quorum slice $Q$, and the PR value of $G$ is denoted by $\text{PR}_{G}.$ 
For ease of understanding NR, we present the following example. 
Assume that only one quorum slice $Q$ includes node $v$; 
$Q=\{t:3, n_1, n_2,\{t:1, v, n_3\}\}$, and $Q$ is created by two nodes whose PR values are 0.01 and 0.02, respectively.
Then $a_{l(Q,v)}(Q,v)$ is 
$$a_{l(Q,v)}(Q,v)=a_{2}(Q,v)=\frac{3}{3}\times a_1(\{t:1, v, n_3\},v)=\frac{1}{2},$$
and NR of $v$ is $0.03\times\frac{1}{2}.$

\begin{figure}[t]
\centering
\includegraphics[width=\columnwidth]{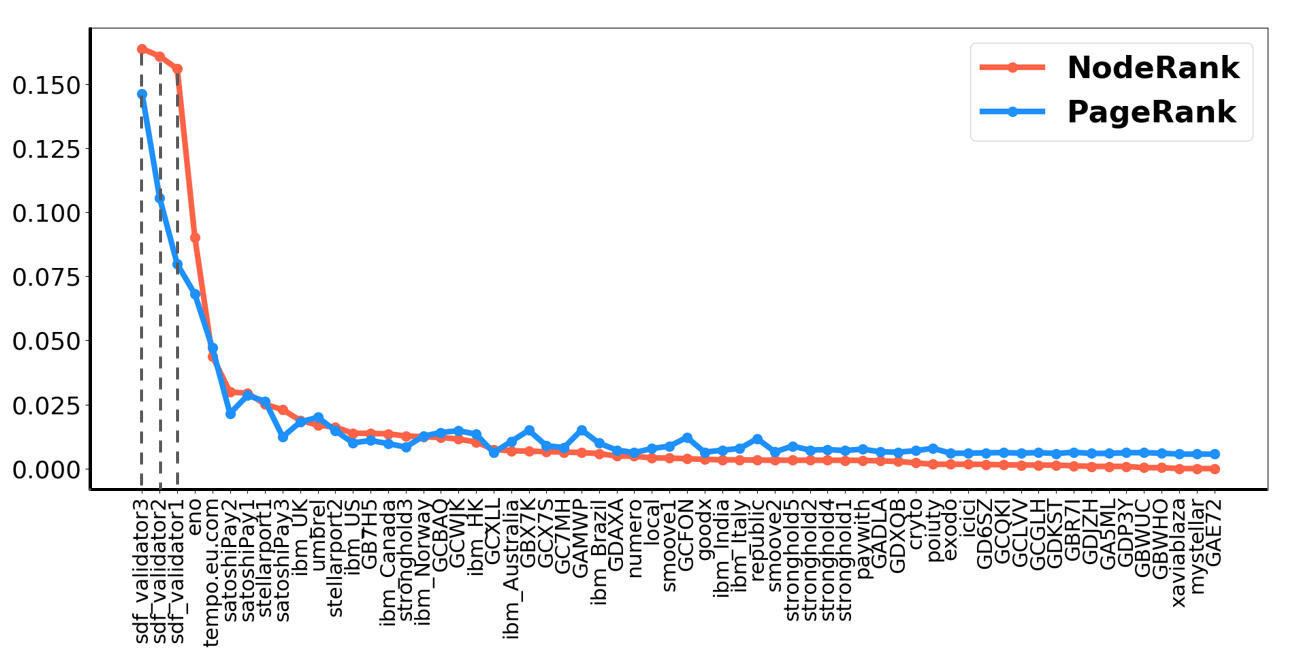}
\caption{NodeRank (NR) and PageRank (PR)}
\label{fig:metrics}
\end{figure}

We calculate NR of the nodes according to Eq.~\eqref{eq:nr}, and normalize the values. 
Fig.~\ref{fig:metrics} shows the result of the two metrics, and the $x$ and $y$-axes represent the names of the validators and values of metrics, respectively.
The red and blue lines represent NR and PR, respectively.
From Fig.~\ref{fig:metrics}, it can be seen that three nodes, namely \textit{sdf\_validator1, sdf\_validator2,} and \textit{sdf\_validator3}, have the largest NR and PR values, implying that they are significantly influential.
However, when comparing their NR and PR values, the difference between the NR values of \textit{sdf\_validators} is small, while the difference between the PR values of \textit{sdf\_validators} is large.
This implies that the influences of \textit{sdf\_validators} are similar in terms of NR.
In Section~\ref{sec:cascade}, we will confirm that \textit{sdf\_validators} have similar influence with respect to cascading failure as the NR result shows.

Consequently, according to these two metrics, we quantitatively show that the node influence is biased toward a few nodes. 
Specifically, the three nodes run by the Stellar foundation have the greatest influence. 

\smallskip\noindent
\textbf{Causal Analysis. }
This problem may also be drawn by the fact that Stellar is based on the trust of users.
Under the trust model, it is inevitable for the node's influence to be biased toward well-known nodes or nodes directly related to many others. 
Through data analysis, we confirm this and observe that at least one of the \textit{sdf\_validator} nodes is included in all slices.
Note that many existing validators are having business with Stellar as shown in Table~\ref{table:purposeOfValidator}.
%
%
%
%

\smallskip\noindent
\textbf{Summary. }In summary, our data analysis indicates that that the number of validators in each slice is significantly small, and that the structure of quorum slices is biased, which indicates that the current Stellar structure of quorum slices is highly centralized.

\section{Cascading Failure}
\label{sec:cascade}

In the previous section, we see that the structure of quorum slices is highly centralized. 
In this section, we study \textit{cascading failure}. 
Note that \textit{cascading failure} in interconnected systems implies a situation where the failure of a few nodes triggers gradual failure of other nodes.
In the case of SCP, there might exist dependent node failures because each node is designed to be influenced by other nodes. This means that if one node fails, it may be possible for another node to fail as an effect of the failed node.
Moreover, in Stellar, the degree of robustness against the cascading failure would depend largely on the structure of quorum slices. 
Thus, the current Stellar system can be highly vulnerable to the cascading failure because of a significantly centralized structure of quorum slices.
To find out how much the system can be currently weakened by the cascading failure, we apply it to Stellar.

\subsection{Cascading Failure in the Stellar system}

We consider the quorum slices obtained from the Stellar system on January 8, 2019, consisting of a total of 62 validators. 
We first define $Failure$ ratio as follows:
\begin{equation*}
    \text{Failure }(\%)= \frac{\text{Number of failed validators}}{\text{Total number of validators}}\times 100.
\end{equation*}
Note that the number of failed validators are equal to the number of nodes in group \textit{A} or \textit{B}, as explained in Section ~\ref{sec:theoretical}.
As a baseline, 18 nodes (29\% of nodes) in the Stellar system were already offline for unknown reasons.
In such a state, we confirm that when one of the online nodes becomes unavailable (i.e., it is forced to go offline owing to a targeted attack, such as DDoS or network attacks~\cite{ddos,bitcoin_ddos,NEO_ddos}), the node with the 18 offline nodes (i.e. (18+1)/62*100=30.6\% of the Stellar system) are unable to reach a consensus.
In this case, one can see that if one node fails, the system is not significantly influenced.

However, the situation is very different when two nodes become unavailable, as Fig.~\ref{fig:fail2} shows.
In the figure, the $x$-axis represents all possible pairs of validators, and the $y$-axis represents $Failure$, the percentage of nodes that becomes unable to reach a consensus. 
Red dots represent cases when the $Failure$ ratio is above 90 \%.
Six such cases are shown in the figure, and summarized as follows:

\begin{itemize}[topsep=2pt,itemsep=0pt,parsep=0pt,partopsep=0pt,leftmargin=18pt]
    \item[(1)] \texttt{\small sdf\_validator1} and \texttt{\small sdf\_validator2}: $Failure$  = 100 \%
    \item[(2)] \texttt{\small sdf\_validator1} and \texttt{\small sdf\_validator3}: $Failure$  = 100 \%
    \item[(3)] \texttt{\small eno} and \texttt{\small sdf\_validator1}: $Failure$ = 90.3 \%
    \item[(4)] \texttt{\small sdf\_validator2} and \texttt{\small sdf\_validator3}: $Failure$  = 100 \%
    \item[(5)] \texttt{\small eno} and \texttt{ \small sdf\_validator2}: $Failure$  = 90.3\%
    \item[(6)] \texttt{\small eno} and \texttt{\small sdf\_validator3}: $Failure$  = 91.9 \%
\end{itemize}
The results show that the failure of two nodes with large PR and NR values causes large portions of the system to fail.

To observe how cascading failure proceeds, we describe an example, where \textit{sdf\_validator1} and \textit{sdf\_validator2} fail (case 1).
Initially (round 1), 20 nodes fail to reach a consensus. Note that 18 nodes were already offline before deleting the two nodes. The failure of these 20 nodes causes an additional 28 nodes (round 2) to fail (reaching a total of 48 nodes), and those 48 nodes lead to failure (round 3) of an additional 12 nodes (reaching a total of 60 nodes). 
Finally, the failure of these 60 nodes causes the remaining 2 nodes to fail (reaching a total of 62 nodes) over two rounds (round 4 and 5). 
Because the structure of quorum slices in Stellar consists of 62 nodes (validators), all nodes fail to reach the consensus over five rounds, when deleting only two nodes.

\begin{figure}[t]
\centering
\includegraphics[width=\columnwidth]{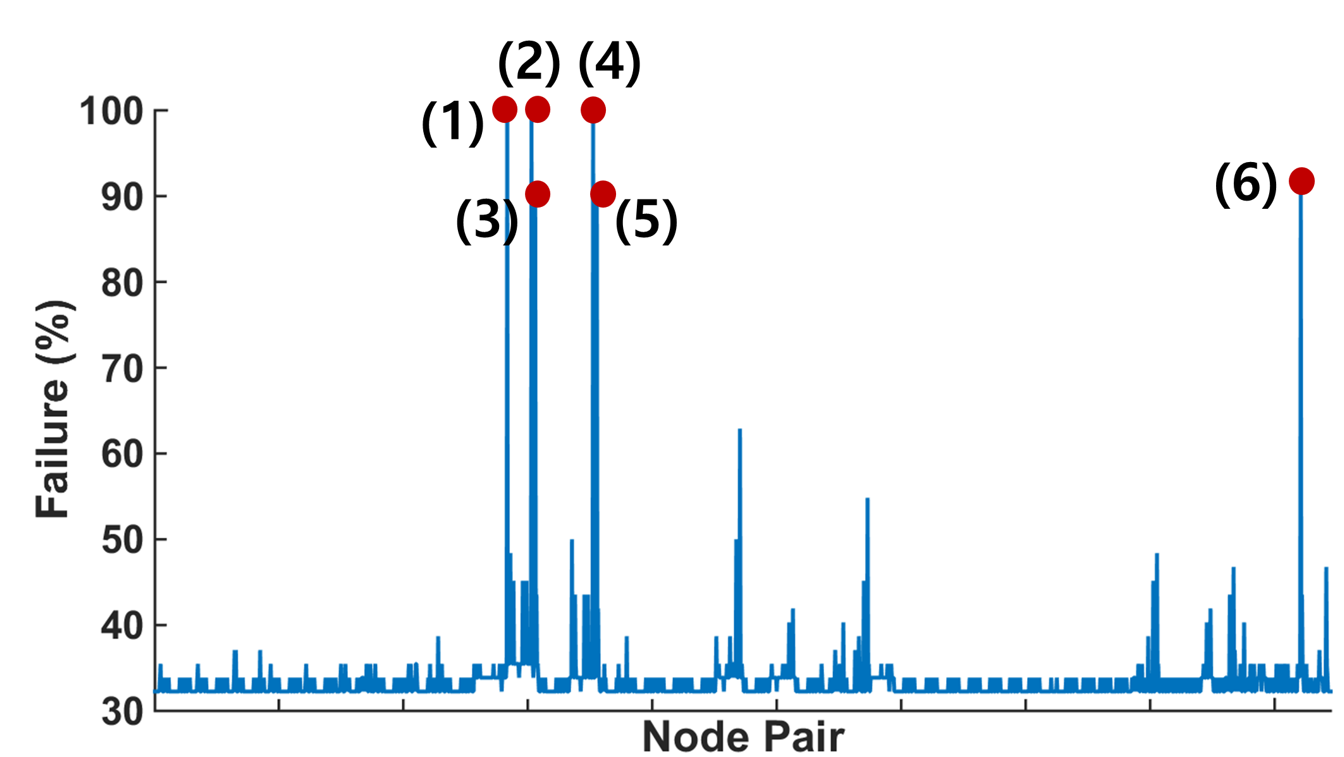}
 \caption{Results of the cascading failure}
 \label{fig:fail2}
\end{figure}

We note that the $Failure$ is 100\% in cases (1), (2), and (4).
This shows that under the current Stellar protocol, if only two nodes are ill-behaved (i.e., two nodes go to group \textit{A}), then all the other nodes eventually cannot reach consensus, because they all go to group \textit{B}.
To make matters worse, the two nodes are run by the Stellar foundation, which implies that \textit{the Stellar system is vulnerable to a single point of failure}. 

In Section~\ref{sec:theoretical}, we defined a $(f,x)$-FT system and showed that the maximum value of $x$ in FBA is $\frac{100}{3}$.
Considering the result of cascading failure in the Stellar system, one can see that Stellar is currently a $(2, \frac{50}{11})$-FT system (i.e., $\frac{2}{44}\times100$).
Note that although there are a total of 62 validators, 18 nodes among them were offline; thus, when calculating $x$, we consider only 44 who are active validators.
As compared with the best case in the FBA model (PBFT whose $x$ value is $\frac{100}{3}$), we can see that Stellar is much more vulnerable to node failures. 
\section{Discussion}
\label{sec:discuss}

\subsection{Decentralized Structure of Quorum Slices}

Because the Stellar system is highly dependent on the structure of the quorum slices, a decentralized structure is required for Stellar to be secure.
However, we see that it is highly centralized, based on our analysis result in Section~\ref{sec:analysis}.

To make the structure of Stellar decentralized, the number of validators should be sufficiently large, and all validators should be included in different slices evenly.
However, it is hard to increase the number of validators because currently, there is no internal incentivization.
Only a small number of nodes that can benefit from external benefits are operating as validators. 
Moreover, it is hard for all validators to be evenly trusted by others because users tend to choose only a  few validators that are popular or in partnership with him/her.
As a result, unless there is no internal incentivization system, it would be hard for the structure of quorum slices to be more decentralized because of lack of participants.
Furthermore, even if there are many validators somehow, we cannot ensure that users would include the validators equally in their slice by nature of a trusting relationship.
For these reasons, we expect that it is not easy to make the structure of quorum slices in Stellar decentralized.

\subsection{Mitigation of Cascading Failures}

In this section, we describe potential solutions to reduce the impact of cascading failures and their limitations.

One solution is making the quorum slice structure of Stellar similar that of PBFT because the maximum value of $x$ in Stellar is that of PBFT.
However, there are some limitations.
First, to be the same with the PBFT style, every user is forced to have the same slice; this would not be accepted in Stellar, which claims that each user can have any desired quorum slice.
Second, unlike Stellar, where users select their validators manually, PBFT allows the system itself to rearrange the slice dynamically and securely depending on who is Byzantine.
Surely, if every user always monitors messages from selected validators and changes its slice immediately whenever the validators are considered to have failures, then Stellar may work as a PBFT style.
However, it is not practical to do so because it requires that all users perform monitoring every minute.
In fact, to see how often users change their quorum slices, we observed it every week for a month, and determined that 53 out of 62 validators in the system (85.5\%) had never changed their slices.
This suggests that most users tend not to change slices.
As a result, it may be hard to force users to change their slices dynamically and securely as a PBFT system.

Users can lower the threshold value of a quorum slice, which would certainly increase the liveness of Stellar. 
However, this would also decrease the safety.
Indeed, it would be complicated to find the optimal threshold value to satisfy both safety and liveness at the same time.

If many popular and important financial institutions participated in SCP while advertising their nodes to others, Stellar could have better liveness because a user could choose many diverse validators, decentralizing the quorum slice structure.
Nevertheless, attracting such institutions to Stellar is still challenging. 
\section{Related works}
\label{sec:related}

There are still limitations in applying PBFT in a public blockchain.
For example, because PBFT needs a large number of communications with nodes for the consensus, it is unsuitable for a public blockchain that has to make consensus with large groups.
Furthermore, if PBFT is used in a public blockchain, it would be vulnerable to Sybil attacks.
Despite these limitations, because PBFT has advantages such as high transaction throughput and reduced waste of energy, there have been many attempts to use PBFT in a public blockchain consensus algorithm.
In 2014, Kwon et al. proposed Tendermint~\cite{kwon2014tendermint}, which is a combination of proof-of-stake (PoS) and BFT.
In Tendermint, nodes vote for validators and the voting power is proportional to the stake deposited by voters.
A block is committed when more than 2/3 of validators agree on the same block.
In 2017, Gilad et al. introduced Algorand~\cite{gilad2017algorand}, which scales the consensus to many users using the verifiable random function (VRF) and assigns each node weight based on the money the node has in its account. 
Algorand prevents Sybil attacks by presenting more chances for weighted users to be a member of round leaders.
In addition, it prevents targeted attacks using the VRF by making the attacker cannot predict the random number that is used for electing the next round leader.
Additionally, many other public blockchain consensus algorithms based on the PBFT exist, such as HoneyBadgerBFT~\cite{miller2016honey}, Zilliqa~\cite{ziliqa}, and NEO~\cite{neo}.
Further, Ripple~\cite{armknecht2015ripple} uses PBFT in a private blockchain, and Nimble~\cite{innerbichler2018federated} uses SCP in creating cloud manufacturing platforms.

Bitcoin or Ethereum have been studied from various perspectives, such as analyzing the information propagation process, mining and anonymity~\cite{decker2013information, miller2015discovering, eyal2014majority, sapirshtein2016optimal, kwon2017selfish, kim2018measuring, reid2013analysis, moser2013anonymity}.
On the other hand, there is a relatively small number of studies on other cryptocurrencies~\cite{miller2017empirical}.

As a measurement study, Gencer et al. evaluated the network of Bitcoin and Ethereum using several metrics, and demonstrated that the network is centralized~\cite{gencer2018decentralization}.
While ~\cite{gencer2018decentralization} focused on an analysis of Bitcoin and Ethereum, which are the PoW-based systems, we analyzed another consensus protocol in this study.
Moreover, we studied both the seriousness of the Stellar system's centralization among validators and its impact on the entire system.
Moreover, in contrast to our negative result, another work~\cite{garcia2018federated} discovered that the correctness of SCP is stronger than the one described in the Stellar whitepaper~\cite{mazieres2015stellar}.
However, they assumed a moderately secure structure of quorum slices.
In contrast, in this study, we described that the security of the Stellar system is significantly dependent on the structure of the quorum slices, and analyzed the weakness of liveness in the Stellar system, considering the current structure.

In fact, the concept of a quorum has existed for a long time in distributed computing.
It presents the minimum number of votes on a given transaction required for a consistent operation in a distributed system.
The quorum-based system has been used in various fields, such as a replicated database~\cite{gifford1979weighted, herlihy1986quorum}, algorithms for distributed mutual exclusion~\cite{agrawal1989efficient}, and protocols~\cite{jiang2005quorum}.
Additionally, studies for Byzantine quorum systems have been done for consistency and data availability in a distributed system upon Byzantine failures~\cite{malkhi1998byzantine, alvisi2000dynamic, malkhi2000load}.
Nevertheless, it is difficult to apply these previous studies to Stellar because they make different assumptions.
Unlike previous studies, which assume a fixed set of nodes in a quorum, Stellar involves changing sets of validators over time by allowing any node to participate in Stellar.

\section{Conclusion}
\label{sec:conclusion}
For the first time, we analyzed the Stellar system, which is currently in the top ten cryptocurrencies and is based on a trust model with open membership.
Through our FBA analysis, we proved that FBA is not better than PBFT in terms of safety and liveness.
Especially in Stellar, which is a construction for FBA, we analyzed the structure of the current quorum slices and measured the centrality of the system using two metrics: PR and NR, where NR is proposed as a new metric.
These two metrics demonstrate that the system is highly centralized in a few specific nodes, three of which are operated by the Stellar foundation itself. 
In addition, we found that a cascading failure has a significant impact on the Stellar system.
In fact, the entire system can fail completely in sequence if only the two nodes operated by the Stellar foundation are deleted.


\bibliographystyle{IEEEtran}
\bibliography{references}

\begin{appendices}
\end{appendices}

\end{document}